\documentclass[conference]{IEEEtran}
\usepackage{amsmath,cite,amsfonts,amssymb,psfrag}
\usepackage{graphicx}
\usepackage{rotating}
\usepackage{citesort}
\usepackage{color}

\newtheorem{theorem}{Theorem}

\newtheorem{remark}{Remark}
\newtheorem{definition}{Definition}

\newtheorem{lemma}{Lemma}

\newtheorem{proof}{Proof}

\begin{document}
\title{Secret key-based Authentication with a Privacy Constraint}
\IEEEoverridecommandlockouts
\author{
\authorblockN{Kittipong Kittichokechai and Giuseppe Caire\\
\authorblockA{Technische Universit\"{a}t Berlin}
}
}
\maketitle
%%%%%%%%%%%%%%%%%%%%%%%%%%%%%%%%%%%%%%%%%%%%%%%%%%%%%%%%%%%%%%%%%%%%%%%%%%%%%%%%%%%%%%%%%%%%%%%%%%%%%%%%%%%%%%%%%%%
\begin{abstract}
 We consider problems of authentication using secret key generation under a privacy constraint on the enrolled source data.
 An adversary who has access to the stored description and correlated side information tries to deceive the authentication as well as learn about the source.
  We characterize the optimal tradeoff between the compression rate of the stored description, the leakage rate of the source data,
  and the exponent of the adversary's maximum false acceptance probability.
  The related problem of secret key generation with a privacy constraint is also studied where
  the optimal tradeoff between the compression rate, leakage rate, and secret key rate is characterized.
  It reveals a connection between the optimal secret key rate and security of the authentication system.
\end{abstract}
%%%%%%%%%%%%%%%%%%%%%%%%%%%%%%%%%%%%%%%%%%%%%%%%%%%%%%%%%%%%%%%%%%%%%%%%%%%%%%%%%%%%%%%%%%%%%%%%%%%%%%%%%%%%%%%%%%%
\section{Introduction}\label{sec:Chap_ENDUSER_introduction}

We consider the problem of authentication based on secret key generation.
In the enrollment stage, a user provides the source sequence $X^n$ to the system. The source
is compressed into a description $M$ which is stored as a helping message.
Meanwhile, the secret key message $S$ is generated based on the source and will be used
as a reference for authentication.
In the authentication stage, the user provides an authentication sequence $Y^n$ which could be a noisy measurement of the enrolled source sequence. Based on $M$ and $Y^n$,
the secret key is estimated as $\hat{S}$ and compared with
the reference $S$. The user is successfully authenticated if $\hat{S}=S$.

The system described above can be relevant in several applications including those involving access control, secure, and trustworthy communication. One important class of potential applications is related to using biometric data such as fingerprint, iris scans, and DNA sequences for authentication (see, e.g., \cite{RaneSurvey} and references therein). Unlike passwords, the biometric data inherently belong to users and provide a convenient and seemingly more secure way for authentication. However, it is crucial that privacy of the enrolled data must be protected from any inference of an adversary. The privacy risk in this case is of potentially high impact since the biometric data is commonly tied to the person identity. If it is compromised, it cannot be reverted or changed like in the case of using passwords.
\begin{figure}[t]
    \centering
    \psfrag{x}[][][0.9]{\small{$X^{n}$}}
    \psfrag{y}[][][0.9]{\small{$Y^{n}$}}
    \psfrag{z}[][][0.9]{\small{$Z^{n}$}}
    \psfrag{s}[][][0.9]{\small{$S$}}
    \psfrag{m}[][][0.9]{\small{$M,\ \color{blue}\text{rate}\ R$}}
    \psfrag{shat}[][][0.9]{\small{$\hat{S},\ \color{blue}\text{Pr}(\hat{S} \neq S) \leq \delta$}}
    \psfrag{enc}[][][0.9]{\small{Encoder}}
    \psfrag{dec}[][][0.9]{\small{Decoder}}
    \psfrag{adv}[][][0.9]{\small{Adversary}}
    \psfrag{yt}[][][0.9]{\small{$\tilde{y}^{n}$}}
    \psfrag{L}[][][0.9]{\small{\color{blue}$\frac{1}{n}I(X^n;M,Z^n) \leq L+\delta$}}
    \psfrag{stilde}[][][0.9]{\small{$\tilde{S}_{\tilde{y}^n}$}}
    \psfrag{fap}[][][0.9]{\small{$\color{blue}\text{Pr}(\tilde{S}_{\tilde{y}^n} = S) \leq 2^{-n(E-\delta)}$}}
    \includegraphics[width=7.5cm]{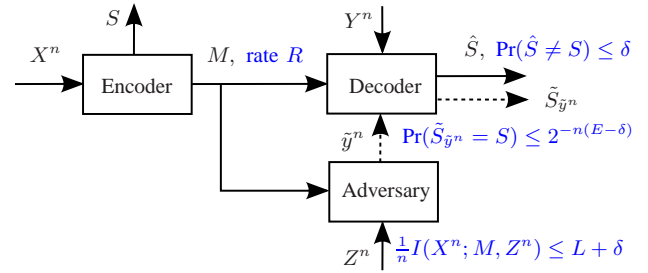}
    \caption{Secret key-based authentication system with a privacy constraint.}
    \vspace{-0.2cm}
    \label{fig:mFAP}
\end{figure}

In this work, we consider the secret-key based authentication problem in the presence of an adversary, who has access to the stored description $M$ as well as correlated side information $Z^n$, as shown in Fig. \ref{fig:mFAP}.
The adversary tries to deceive the authentication using its own sequence $\tilde{y}^n$ and is also interested in learning about the enrolled source data $X^n$. We call the event where the legitimate user fails during the authentication as a false rejection, and the event where the system accepts the adversary as a false acceptance. As for the privacy constraint, normalized mutual information between the enrolled source data $X^n$ and all information available at the adversary, e.g., $(M,Z^n)$, is used as a measure of information leakage rate. We wish to design an authentication system that achieves negligible false rejection probability and at the same time minimizes 1) the compression rate of the stored description, 2) the leakage rate of the enrolled source, and 3) the maximum false acceptance probability (mFAP) \emph{exponentially}. In general, there exists a tradeoff between the compression rate, the information leakage rate, and the mFAP exponent. For example, to obtain a large mFAP exponent while achieving reliable authentication for the legitimate user, a ``high quality" description $M$ may need to be stored which in turn can lead to high amount of information leakage. The main result of this work is a single-letter characterization of the fundamental tradeoff between the compression rate, information leakage rate, and mFAP exponent for discrete memoryless sources.

Closely related to the setting described above, we consider also the problem of secret key generation (for authentication) with a privacy constraint where, apart from reliable reconstruction of the secret key, we wish to maximize the secret key rate as well as ensuring that the leakage rate of the key is negligible. Also in this case, the optimal tradeoff between the compression rate, leakage rate of the source, and secret key rate is characterized.
In particular, the optimal secret key rate is shown to be equivalent to the optimal mFAP exponent derived in the first problem.

\vspace{-0.06cm}
\subsection*{Related Work}
Authentication problems from an information theoretic perspective have been studied in several directions. Maurer in \cite{Maurer} considered the message authentication problem in connection with the hypothesis testing problem where the underlying message probability distributions of the legitimate user and adversary are assumed to be different. Martinian et al. \cite{Martinian} considered authentication with a distortion criteria. More recently, works appear to consider authentication problems based on secret key generation \cite{AhlswedeCsiszar}. These include for example works \cite{IgnatenkoWillems},\cite{WillemsIgnatenko},\cite{LaiHoPoor} which focused on biometric authentication systems
 where privacy of the enrolled data is also taken into account. Analysis of deception probability in the authentication system from an adversary's perspective was also considered in \cite{Kang}. Closely related to the secret key-based authentication problem with privacy constraint are the problems of source coding with privacy constraint, e.g.,  \cite{VillardPiantanida},\cite{Kittichokechai}, where the goals are to reconstruct the source reliably while preserving the privacy of the source or the reconstruction sequences from any inference of an eavesdropper. In this work, we extend the problem in \cite{WillemsIgnatenko} to a more general case where the adversary has correlated side information. Moreover, we provide a complete characterization of the problem studied in \cite{LaiHoPoor}. Standard notations in \cite{ElGamalKim} are used.
%%%%%%%%%%%%%%%%%%%%%%%%%%%%%%%%%%%%%%%%%%%%%%%%%%%%%%%%%%%%%%%%%%%%%%%%%%%%%%%%%%%%%%%%%%%%%%
\section{Secret key-based Authentication System}

\subsection{Problem Formulation}\label{sec:problem_setting_mFAP}
Let us consider a secret key-based authentication system  shown in Fig. \ref{fig:mFAP}. Source and side information alphabets, $\mathcal{X}, \mathcal{Y}, \mathcal{Z}$ are assumed to be finite. Let $(X^{n},Y^n,Z^n)$ be $n$-length sequences which are i.i.d. according to $P_{X,Y,Z}$.%, i.e., $\mathrm{Pr}(X^n=x^n,Y^n=y^n,Z^n=z^n)=\prod_{i=1}^n P_{X,Y,Z}(x_i,y_i,z_i)$.

 In the enrollment stage, based on the user's source sequence $X^n$, an ``encoder" generates a rate-limited description $M \in \mathcal{M}^{(n)}$
 and a secret key message $S \in \mathcal{S}^{(n)}$.
 For authentication, the user provides a (noisy) authentication sequence $Y^n$ to the system. Based on $Y^n$ and the stored description $M$, a ``decoder" generates $\hat{S}$ as an estimate of the secret key. The user will be positively authenticated if $\hat{S}=S$.

 The information leakage rate at the adversary who has access to the stored description $M$ and
 %has its own
 side information $Z^n$, correlated with $X^n$,
 is measured by the normalized mutual information $I(X^n;M,Z^n)/n$. The adversary,
 based on $M$ and $Z^n$, also chooses a sequence $\tilde{y}^n(M,Z^n) \in \mathcal{Y}^n$ for authentication.
 The maximum false acceptance probability (mFAP) is defined as $\text{mFAP} \triangleq \max_{\tilde{y}^n(M,Z^n) \in \mathcal{Y}^n} \text{Pr}(\tilde{S}_{\tilde{y}^n}=S)$, where $\tilde{S}_{\tilde{y}^n}$ is the estimate resulting from $M$ and $\tilde{y}^n$.
 We are interested in characterizing the optimal tradeoff between the compression rate, information leakage rate, and mFAP exponent.

%Definitions of code, achievability, and the compression-leakage-mFAP exponent region are given below.
\begin{definition}\label{def:code_mFAP}
A code for secret key-based authentication with a privacy constraint consists of
\begin{itemize}
  \item an encoder $f_m^{(n)}: \mathcal{X}^{n} \rightarrow \mathcal{M}^{(n)}$,
  \item an encoder $f_s^{(n)}: \mathcal{X}^{n} \rightarrow \mathcal{S}^{(n)}$,
  \item a decoder $g^{(n)}: \mathcal{M}^{(n)} \times \mathcal{Y}^{n} \rightarrow \mathcal{S}^{(n)}$,
\end{itemize}
where  $\mathcal{M}^{(n)}$ and $\mathcal{S}^{(n)}$ are  finite sets.
\end{definition}

\begin{definition}  A compression-leakage-mFAP exponent tuple $(R,L,E) \in \mathbb{R}^{3}_{+}$ is said to be \emph{achievable} if for any $\delta>0$ and all sufficiently large $n$ there exists a code above such that
\begin{align}
\text{Pr}(\hat{S}\neq S)&\leq \delta, \label{eq:FRPconstraint}\\
\frac{1}{n}\log\big|\mathcal{M}^{(n)}\big| &\leq R+\delta,\label{eq:rate_constraint}
\end{align}
\begin{align}
\frac{1}{n}I(X^{n};M,Z^n) &\leq L+\delta,\label{eq:leakage_constraint}\\
\text{and} \ \  \frac{1}{n}\log \frac{1}{\text{mFAP}} &\geq E-\delta. \label{eq:mFAPconstraint}
\end{align}
The \emph{compression-leakage-mFAP exponent} region $\mathcal{R}_1$ is the set of all achievable tuples.
\end{definition}

\subsection{Result}

\begin{theorem}\label{theorem:region_exponent}%[Compression-leakage-mFAP exponent region]\label{theorem:region_exponent}
The compression-leakage-mFAP exponent region $\mathcal{R}_1$ for the problem depicted in Fig. \ref{fig:mFAP} is given by a set of all tuples $(R,L,E)\in \mathbb{R}^{3}_{+}$ such that
\begin{align}
R &\geq I(X;V|Y), \\
L & \geq I(X;V,Y) - I(X;Y|U)+I(X;Z|U),\\
E &\leq I(V;Y|U)-I(V;Z|U),
\end{align}
for some joint distributions of the form $P_{X,Y,Z}P_{V|X}P_{U|V}$
with  $|\mathcal{U}| \leq |\mathcal{X}| +3, |\mathcal{V}| \leq (|\mathcal{X}|+3)(|\mathcal{X}|+2)$.
\end{theorem}

\begin{remark}[Randomized encoder]
Theorem \ref{theorem:region_exponent} holds also for a more general setting which allows randomized encoders, i.e., $M$ and $S$ are randomly generated according to $p(m|x^n)$ and $p(s|x^n)$, respectively. This can be seen from the converse proof of Theorem \ref{theorem:region_exponent} that no assumption regarding the deterministic encoders was made.
\end{remark}

\begin{remark}[Special cases]\label{remark:special_cases}\  \par
i) When side information at the adversary is degraded, i.e., $X-Y-Z$ forms a Markov chain, the compression-leakage-mFAP exponent region is reduced to the set $\mathcal{R}_{1,X-Y-Z}$ consisting of all tuples $(R,L,E)$ such that
\begin{align*}
R &\geq I(X;V|Y), \\
L & \geq I(X;Z)+I(X;V|Y),\\
E &\leq I(V;Y|Z),
\end{align*}
for some joint distributions of the form $P_{X,Y}P_{Z|Y}P_{V|X}$. We obtain this region from $\mathcal{R}_{1}$ by setting $U$ constant. %in theorem \ref{theorem:region_exponent}.
The converse proof is modified slightly and is provided in Appendix \ref{appendix:converse_degraded_case}.

ii) When the adversary has no side information, the result in Theorem \ref{theorem:region_exponent} reduces to that in \cite{WillemsIgnatenko}. For example, by setting $Z$ and $U$ equal to constants and $R=H(X)$, we recover \cite[Theorem 4]{WillemsIgnatenko}.
\end{remark}

\begin{proof}[Proof of Theorem \ref{theorem:region_exponent}]
 The sketch of achievability proof is given below based on a random coding argument where we use the definitions and properties of $\epsilon$-typicality as in \cite{ElGamalKim}. Our achievable scheme utilizes  layered coding and binning,
 while the converse proof  for the information leakage rate is inspired by that of the secure source coding problem \cite{VillardPiantanida}.

 \textbf{Achievability}: Fix $P_{V|X}$ and $P_{U|V}$. Let $\epsilon$ and $\delta_{\epsilon}$ be positive real numbers where $\delta_{\epsilon} \rightarrow 0$ as $\epsilon \rightarrow 0$. Assume that $I(V;Y|U)-I(V;Z|U) >0$. The case where $I(V;Y|U)-I(V;Z|U) \leq 0$ is trivial since the encoder can just set the secret key message to be constant and does not transmit
 at all, implying that $(R,L,E)=(0,I(X;Z),0)$ is achievable.

 1) \emph{Codebook generation:} Randomly and independently generate $2^{n(I(X;U) + \delta_{\epsilon})}$ $u^n(j)$ sequences, each i.i.d. according to $\prod_{i=1}^nP_U(u_i)$, $j \in [1:2^{n(I(X;U) + \delta_{\epsilon})}]$. Then distribute them uniformly at random into $2^{n(I(X;U|Y) + 2\delta_{\epsilon})}$
 bins $b_U(m_1)$, $m_1 \in [1:2^{nI(X;U|Y) + 2\delta_{\epsilon}}]$. For each $j$, randomly and conditionally independently generate $2^{n(I(X;V|U) + \delta_{\epsilon})}$ $v^n(j,k)$ sequences, each i.i.d. according to $\prod_{i=1}^nP_{V|U}(v_i|u_i)$, $k \in [1:2^{n(I(X;V|U) + \delta_{\epsilon})}]$, and distribute these sequences uniformly at random into $2^{n(I(X;V|U,Y) + 3\delta_{\epsilon})}$
 bins $b_V(j,m_2)$, $m_2 \in [1:2^{nI(X;V|U,Y) + 3\delta_{\epsilon}}]$. Moreover, in each bin $b_V(j,m_2)$, we distribute sequences $v^n$ uniformly at random into subbins, indexed by $s$, where $s \in [1:2^{n(I(V;Y|U)-I(V;Z|U)-\delta_{\epsilon})}]$. The index $s$ here represents a subbin index of the second-layered bin. In each subbin, there are $2^{n(I(V;Z|U)-\delta_{\epsilon})}$ sequences $v^n$, each indexed by $s'$. Note that $k=(m_2,s,s')$ here. The codebooks are then revealed to all parties.

 2) \emph{Enrollment:} Given $x^n$, the encoder looks for $u^n(j)$ and $v^n(j,k)$ that are jointly typical with $x^n$. From the covering lemma \cite{ElGamalKim}, with high probability, there exist such codeword pairs. If there are more than one pairs, the encoder selects one of them uniformly at random,
 and then sends the corresponding bin indices $m_1$ and $m_2$ to the decoder. The total rate is thus equal to $I(X;U|Y)+I(X;V|U,Y) + 5\delta_{\epsilon} = I(X;V|Y) + 5\delta_{\epsilon}$. The secret key is set to be the subbin index $s$ in which the chosen sequence $v^n \in  b_V(j,m_2)$ falls.

 3) \emph{Authentication:} The decoder looks for $u^n(j)$ and $v^n(j,k)$ in the bins $(m_1,m_2)$ which are jointly typical with $y^n$. From the packing lemma \cite{ElGamalKim}, with high probability, it will find the unique sequence $u^n(j) \in b_U(m_1)$ which is jointly typical with $y^n$. Then, with high probability, it will find the unique $v^n(j,k) \in b_V(j,m_2)$ which is jointly typical with $y^n$ and the decoded $u^n(j)$. Finally, it puts out the corresponding subbin index of the decoded $v^n$ as an estimate of the secret key which, with high probability, will be equal to the generated one.

Let $U^n(J)$ and $V^n(J,K)$ be the codewords chosen at the encoder in the enrollment stage, and $(M_1,M_2)$ be the corresponding  indices of the bins to which $U^n(J)$ and $V^n(J,K)$ belong. Note that $(M_1,M_2)$ can be determined from $(J,K)$.

From the enrollment stage, the sources and selected codewords are jointly typical, i.e.,  $(X^n,U^n(J),V^n(J,K),Y^n,Z^n) \in \mathcal{T}_{\epsilon}^{(n)}$, with high probability. We have the following lemma.
\begin{lemma} \label{lemma:1}
The following bound holds, $H(Z^n|J) \leq n(H(Z|U)+\delta_{\epsilon})$.
\end{lemma}
\begin{proof}
The proof is given in Appendix \ref{appendix:Lemma1}.
\end{proof}

Then, the information leakage averaged over all possible codebooks can be bounded as follows.
\begin{align*}
&I(X^n;M_1,M_2,Z^n) = H(X^n)- H(X^n|M_1,M_2,Z^n)\\
&\leq nH(X)- H(X^n|J,Z^n)+H(M_2)\\
&\leq nH(X)- H(X^n,Z^n) + H(J) +H(Z^n|J)+H(M_2)\\
&\overset{(a)}{\leq} -nH(Z|X) + n(I(X;U)+\delta_{\epsilon}) + n(H(Z|U)+\delta_{\epsilon})\\ & \qquad +n(I(X;V|U,Y)+3\delta_{\epsilon})\\
&\overset{(b)}{\leq} n(I(X;U,Z)+ I(X;V|U,Y)+ \delta_{\epsilon}')\\
&\overset{(c)}{=} n(I(X;V,Y)-I(X;Y|U)+I(X;Z|U)+ \delta_{\epsilon}')\\
& \leq n(L + \delta_{\epsilon}'),
\end{align*}
if $L \geq I(X;V,Y)-I(X;Y|U)+I(X;Z|U)$,
 where $(a)$ follows from the memoryless property of the sources, from the codebook generation, and from bounding the term $H(Z^n|J)$ as in Lemma \ref{lemma:1}, $(b)$ from the Markov chain $U-X-Z$ for some $\delta_{\epsilon}' \geq 5\delta_{\epsilon}$, and $(c)$ from the Markov chain $U-V-X-(Y,Z)$.

As for an achievable mFAP exponent, we consider the adversary who
knows $m=(m_1,m_2)$ and side information $z^n$ and tries to select a sequence $\tilde{y}^n(m,z^n)$ that results in the estimated key $\tilde{S}_{\tilde{y}^n}$ equal to the original key $S$ of the person it claims to be.
From our achievable scheme, the secret key $S$ is chosen from the subbin index of the selected codeword $V^n$. Thus, the adversary only needs to consider $\tilde{S}_{\tilde{y}^n}$ that results from sequences $V^n$ which are jointly typical with $X^n$. There are in total $2^{n(I(X;U,V)+2\delta_{\epsilon})}$ such sequences generated.

 Similarly as in \cite{WillemsIgnatenko}, from the binning scheme with uniform bin and subbin index assignment, we have that the joint probability that a description $m$ is selected and a \emph{certain} secret key $s$ is chosen is equal to a total number of jointly typical sequences $v^n$ with corresponding indices $m$ and $s$ divided by a total number of jointly typical sequences $v^n$. That is,
\begin{align}
\mbox{Pr}(M=m,S=s)
&\leq \frac{\Big\lceil\frac{ \mbox{Pr}(M=m)\cdot 2^{n(I(X;U,V)+2\delta_{\epsilon})}}{|\mathcal{S}|}\Big\rceil}{2^{n(I(X;U,V)+2\delta_{\epsilon})}}. \label{eq:joint_prob}
\end{align}

Let $g(\cdot)$ denote the decoding function used for estimating the secret key message in the achievability scheme. Then
\begin{align*}
&\mbox{mFAP} = \max_{\tilde{y}^n(M,Z^n)\in \mathcal{Y}^n} \text{Pr}(\tilde{S}_{\tilde{y}^n} =S)\\
&= \max_{\tilde{y}^n(M,Z^n)\in \mathcal{Y}^n} \text{Pr}(g(M,\tilde{y}^n(M,Z^n)) =S)\\
&\leq \sum_{m=1,\ldots,|\mathcal{M}|}\sum_{z^n} \max_{\tilde{y}^n(m,z^n)\in \mathcal{Y}^n} \text{Pr}(M=m, Z^n=z^n,\\ & \qquad g(m,\tilde{y}^n(m,z^n))=S) \\
&= \sum_{m}\sum_{z^n} \max_{\tilde{y}^n(m,z^n)\in \mathcal{Y}^n} \text{Pr}(M=m, S= g(m,\tilde{y}^n(m,z^n)) )\cdot \\ &\qquad \text{Pr}(Z^n=z^n|M=m, S=g(m,\tilde{y}^n(m,z^n))) \\
&\overset{(a)}{\leq}  \sum_{m}\Big\lceil\frac{ \text{Pr}(M=m)\cdot2^{n(I(X;U,V)+2\delta_{\epsilon})}}{|\mathcal{S}|}\Big\rceil\cdot \frac{1}{2^{n(I(X;U,V)+2\delta_{\epsilon})}}\\
& \leq \sum_{m} \Big(\frac{ \text{Pr}(M=m)\cdot2^{n(I(X;U,V)+2\delta_{\epsilon})}}{|\mathcal{S}|}+1\Big)\cdot \\ & \qquad \frac{1}{2^{n(I(X;U,V)+2\delta_{\epsilon})}}\\
&\overset{(b)}{=} 2^{-n(I(V;Y|U)-I(V;Z|U)-\delta_{\epsilon})}  + 2^{-n(I(V;Y)-3\delta_{\epsilon})}\\
&\overset{(c)}{\leq} 2^{-n(I(V;Y|U)-I(V;Z|U)-\delta_{\epsilon}')},
\end{align*}
where $(a)$ follows from the uniform bin and subbin index assignment in the achievable scheme and the bound in $\eqref{eq:joint_prob}$, $(b)$ follows from the code construction where $|\mathcal{S}|=2^{n(I(V;Y|U)-I(V;Z|U)-\delta_{\epsilon})}$ and $|\mathcal{M}|=|\mathcal{M}_1||\mathcal{M}_2|=2^{n(I(X;V|Y) + 5\delta_{\epsilon})}$, and $(c)$ follows from the Markov chain $U-V-Y$ which results in $I(V;Y)\geq I(V;Y|U)$.

That is, we have
\begin{align*}
 &\frac{1}{n}\log\frac{1}{\mbox{mFAP}} \geq I(V;Y|U)-I(V;Z|U)-\delta_{\epsilon}' \geq E-\delta_{\epsilon}',
\end{align*}
if $E \leq I(V;Y|U)-I(V;Z|U)$.

\textbf{Converse}: Let $U_i \triangleq (M,Y_{i+1}^n,Z^{i-1})$ and $V_i \triangleq (M,S,Y_{i+1}^n,Z^{i-1})$ which satisfy $U_i-V_i-X_i-(Y_i,Z_i)$ for all $i=1,\ldots,n$ as $U_i$ is included in $V_i$ and $(Y_i,Z_i)$ is independent of $V_i$ given $X_i$ due to the memoryless property of the side information channel $P_{Y,Z|X}$. For any achievable tuple $(R,L,E) \in \mathbb{R}^3_{+}$, it follows that
\begin{align*}
&n(R + \delta_n) \geq H(M) \geq H(M|Y^n)-H(M,S|X^n,Y^n,Z^n)\\
&= H(M,S|Y^n)-H(S|M,Y^n) -H(M,S|X^n,Y^n,Z^n)\\
&\overset{(a)}{\geq}  I(M,S;X^n,Z^n|Y^n)-n\epsilon_n\\
%&= \sum_{i=1}^n H(X_i,Z_i|Y_i)- H(X_i,Z_i|M,S,X^{i-1},Z^{i-1},Y^n)\\ &\qquad -n\epsilon_n\\
&\overset{(b)}{\geq} \sum_{i=1}^n H(X_i,Z_i|Y_i)- H(X_i,Z_i|V_i,Y_i)-n\epsilon_n\\
&\geq \sum_{i=1}^n I(X_i;V_i|Y_i)-n\epsilon_n,
\end{align*}
where $(a)$ follows from Fano's inequality $H(S|M,Y^n) \leq n\epsilon_n$ and $(b)$ follows from the definition of $V_i$ and that conditioning reduces entropy.

The information leakage can be bounded as follows.
\begin{align*}
&n(L + \delta_n)\geq I(X^n;M,Z^n)
= I(X^n;M,S,Y^n)\\ &\qquad -I(X^n;S|M,Y^n) - I(X^n;Y^n|M) +I(X^n;Z^n|M)  \\
&\overset{(a)}{\geq} I(X^n;M,S,Y^n)-n\epsilon_n  - I(X^n;Y^n|M)+I(X^n;Z^n|M)  \\
&= \sum_{i=1}^n H(X_i)-H(X_i|M,S,X^{i-1},Y^n) - H(Y_i|M,Y_{i+1}^n) \\ &\qquad +H(Y_i|M,Y_{i+1}^n,X^n)+H(Z_i|M,Z^{i-1}) \\ &\qquad - H(Z_i|M,Z^{i-1},X^n)-n\epsilon_n  \\
&\overset{(b)}{\geq} \sum_{i=1}^n H(X_i)-H(X_i|M,S,X^{i-1},Y^n,Z^{i-1})-I(Y_i;X_i)\\ &\qquad + I(Y_i;M,Y_{i+1}^n) +I(Z_i;X_i)-I(Z_i;M,Z^{i-1}) -n\epsilon_n  \\
&\overset{(c)}{\geq} \sum_{i=1}^n I(X_i;M,S,Y_i^n,Z^{i-1}) -I(Y_i;X_i)+I(Z_i;X_i) \\ &\qquad + I(Y_i;M,Z^{i-1},Y_{i+1}^n) -I(Z_i;M,Z^{i-1},Y_{i+1}^n)-n\epsilon_n \\
&\overset{(d)}{=} \sum_{i=1}^n I(X_i;V_i,Y_i) -I(Y_i;X_i|U_i) +I(Z_i;X_i|U_i)-n\epsilon_n,
\end{align*}
where $(a)$ follows from Fano's inequality, $(b)$ follows from the Markov chains $X_i-(M,S,X^{i-1},Y^n)-Z^{i-1}$ and $(Y_i,Z_i)-X_i-(M,Y_{i+1}^n,Z^{i-1},X^{n\setminus i})$, $(c)$ follows from the Csisz\'{a}r's sum identity \cite{CsiszarBook}, $\sum_{i=1}^n I(Y_i;Z^{i-1}|M,Y_{i+1}^n) - I(Z_i;Y_{i+1}^n|M,Z^{i-1}) =0$, $(d)$ follows from the definitions of $U_i$ and $V_i$ and the Markov chain $U_i-X_i-(Y_i,Z_i)$.

Lastly, the bound on mFAP exponent $n(E-\delta_n) \leq \sum_{i=1}^nI(V_i;Y_i|U_i)-I(V_i;Z_i|U_i)$ can be shown similarly as in \cite{WillemsIgnatenko} with some modification. This part of the proof is provided in Appendix \ref{appendix:converse_mFAP}.
%%%%%%%%%%%%%%%%%%%%%%%%
The proof ends with the standard steps for single letterization using a time-sharing random variable and letting $\delta_n, \epsilon_n \rightarrow 0$ as $n\rightarrow \infty$.
The cardinality bounds on the sets $\mathcal{U}$ and $\mathcal{V}$ %in $\mathcal{R}_1$
can be proved using the support lemma \cite{CsiszarBook}, and is shown in Appendix \ref{appendix:cardinality}.
\end{proof}

\subsection{Binary Example}
To demonstrate the derived tradeoff, let us consider a simple binary example of the special case in Remark \ref{remark:special_cases}i). Let $X\sim \text{Bern}(1/2)$, $Y$ is an erased version of $X$ with erasure probability $p$, and $Z$ is an erased version of $Y$ with erasure probability $q$. The region $\mathcal{R}_{1,X-Y-Z}$ in Remark \ref{remark:special_cases}i) reduces to the set of all $(R,L,E)$ such that
\begin{align*}
R &\geq p(1-h(\alpha)), \\
L & \geq (1-q)(1-p) + p(1-h(\alpha)),\\
E &\leq q(1-p)(1-h(\alpha)),
\end{align*}
for some $\alpha \in [0,1/2]$. The proof is given in Appendix \ref{appendix:proof_example}. We can see for example that there is a tradeoff between the mFAP exponent and the leakage rate, i.e., in order to increase the mFAP exponent, we need to allow some more leakage.

\section{Secret Key Generation with Privacy Constraint}
In this section, we consider a related problem setting depicted in Fig. \ref{fig:Rs} where, instead of maximizing the mFAP exponent, we are interested in maximizing the secret key rate generated at the enrollment stage as well as protecting the secret key from any inference of an adversary who has access to the description $M$ and side information $Z^n$. This setting without the compression rate constraint was studied in \cite{LaiHoPoor} where the authors characterized inner and outer bounds to the leakage-key rate region. Moreover, it is closely related to the one-way secret key generation with rate constraint in \cite{CsiszarNarayan}.

\subsection{Problem Formulation}
The problem setting follows similarly as that in Section~\ref{sec:problem_setting_mFAP}, except that the mFAP constraint in \eqref{eq:mFAPconstraint} is replaced by the key rate and key leakage constraints. %That is, we have the following definition.
\begin{figure}[t]
    \centering
    \psfrag{x}[][][0.9]{\small{$X^{n}$}}
    \psfrag{y}[][][0.9]{\small{$Y^{n}$}}
    \psfrag{z}[][][0.9]{\small{$Z^{n}$}}
    \psfrag{s}[][][0.9]{\small{$S,\ \color{blue}\text{rate}\ R_s$}}
    \psfrag{m}[][][0.9]{\small{$M,\ \color{blue}\text{rate}\ R$}}
    \psfrag{shat}[][][0.9]{\small{$\hat{S},\ \color{blue}\text{Pr}(\hat{S} \neq S) \leq \delta$}}
    \psfrag{enc}[][][0.9]{\small{Encoder}}
    \psfrag{dec}[][][0.9]{\small{Decoder}}
    \psfrag{adv}[][][0.9]{\small{Adversary}}
    \psfrag{L}[][][0.9]{\small{\color{blue}$\frac{1}{n}I(X^n;M,Z^n) \leq L+\delta$}}
   \psfrag{Lkey}[][][0.9]{\small{\color{blue}$\frac{1}{n}I(S;M,Z^n) \leq \delta$}}
    \includegraphics[width=7.5cm]{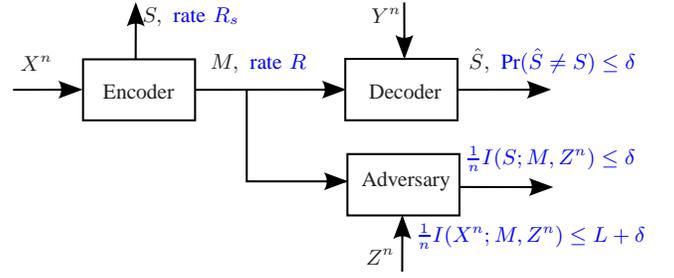}
    \caption{Secret key generation for authentication with a privacy constraint.}
    \label{fig:Rs}
\end{figure}

\begin{definition}  A tuple $(R,L,R_s) \in \mathbb{R}^{3}_{+}$ is said to be \emph{achievable} if for any $\delta>0$ and all sufficiently large $n$ there exists a code consisting of encoders and a decoder (as in Definition~\ref{def:code_mFAP}) such that \eqref{eq:FRPconstraint}-\eqref{eq:leakage_constraint} hold and
\begin{align}
 \frac{1}{n}H(S) \geq R_s & -\delta,\\
 \frac{1}{n} I(S;M,Z^n) &\leq \delta.
\end{align}
The \emph{compression-leakage-key rate region} $\mathcal{R}_2$ is the set of all achievable tuples.
\end{definition}

\subsection{Result}
\begin{theorem}\label{theorem:region_key rate}
The compression-leakage-key rate region $\mathcal{R}_2$ for the problem in Fig. \ref{fig:Rs} is given by a set of all tuples $(R,L,R_s)\in \mathbb{R}^{3}_{+}$ such that
\begin{align}
R &\geq I(X;V|Y), \\
L & \geq I(X;V,Y) - I(X;Y|U)+I(X;Z|U),\\
R_s &\leq I(V;Y|U)-I(V;Z|U),
\end{align}
for some joint distributions of the form $P_{X,Y,Z}P_{V|X}P_{U|V}$
with  $|\mathcal{U}| \leq |\mathcal{X}| +3, |\mathcal{V}| \leq (|\mathcal{X}|+3)(|\mathcal{X}|+2)$.
\end{theorem}

\begin{remark}
Although different achievable schemes were used, the inner bound in \cite{LaiHoPoor} coincides with the compression-leakage-key rate region $\mathcal{R}_2$ where $R=H(X)$. Here we provide the complete result by establishing a matching converse. In addition, the extra compression rate constraint is considered where the layered binning scheme is shown to be optimal.
\end{remark}
\begin{remark}
%We observe that 
The regions specified in Theorems \ref{theorem:region_exponent} and \ref{theorem:region_key rate} have the same form. In particular, the maximum secret key rate in Theorem \ref{theorem:region_key rate} is equal to the maximum mFAP exponent presented in Theorem \ref{theorem:region_exponent}.
Intuitively, this follows from the fact that the coding scheme used to prove Theorem \ref{theorem:region_exponent} also achieves negligible key leakage rate, implying that the adversary has no useful knowledge about the key. It can then only guess the key from possible values in the set $\mathcal{S}$ whose cardinality is at least $2^{H(S)}$.
A similar observation for the case without adversary's side information was noted in \cite{WillemsIgnatenko}.
\end{remark}
\begin{proof}[Proof of Theorem \ref{theorem:region_key rate}]
Proofs for the compression rate $R$ and leakage rate $L$ remain the same as those of Theorem \ref{theorem:region_exponent}. Here we only provide the proof of the secret key rate.

\textbf{Achievability}: With the same achievable scheme as in the proof of Theorem \ref{theorem:region_exponent}, it follows that
 \begin{align*}
H(S) &\geq H(S|J,M_2,S')= H(S,J,M_2,S')-H(J,M_2,S') \\
&\overset{(a)}{\geq} H(U^n,V^n)-H(J)-H(M_2)-H(S')\\
&\overset{(b)}{\geq} n(I(X;U,V)-2\delta_{\epsilon}) - n(I(X;U)+\delta_{\epsilon})\\ & \qquad -n(I(X;V|U,Y)+3\delta_{\epsilon})-(I(V;Z|U)-\delta_{\epsilon})\\
& \geq n(I(Y;V|U)-I(Z;V|U) -\delta_{\epsilon}') \geq n(R_s-\delta_{\epsilon}'),
\end{align*}
if $R_s \leq I(Y;V|U)-I(Z;V|U)$, where $(a)$ follows since $(U^n,V^n)$ are functions of $(J,K)=(J,M_2,S,S')$ given the codebook, and $(b)$ follows from the codebook generation and the properties of jointly typical sequences, i.e., $p(u^n,v^n) \leq \sum_{x^n \in \mathcal{T}_{\epsilon}^{(n)}(X|u^n,v^n)}p(x^n)\leq 2^{-n(I(X;U,V)-2\delta_{\epsilon})}$.

The key leakage averaged over all possible codebooks can be bounded as follows.
 \begin{align*}
&I(S;M_1,M_2,Z^n) %= H(S)-H(S|M_1,M_2,Z^n)\\
\leq H(S)-H(S|J,M_2,Z^n)\\
&= H(S)-H(S,J,M_2,Z^n)+H(J,M_2,Z^n)\\
&\leq H(S) -H(S,J,M_2,Z^n,S')+H(S'|S,J,M_2,Z^n)\\&\qquad +H(J)+H(M_2)+H(Z^n|J)\\
&\overset{(a)}{\leq} H(S)-H(U^n,V^n,Z^n) + n\epsilon_n +H(J)\\&\qquad+H(M_2)+H(Z^n|J)
\end{align*}
 \begin{align*}
&\overset{(b)}{\leq} H(S) -n(I(X;U,V)+H(Z|U,V)-3\delta_{\epsilon})\\&\qquad + n\epsilon_n + n(I(X;U)+\delta_{\epsilon})+n(I(X;V|U,Y)+3\delta_{\epsilon}) \\ &\qquad +n(H(Z|U)+\delta_{\epsilon}) \\
&\overset{(c)}{\leq} n\delta_{\epsilon}'',
\end{align*}
where $(a)$ follows since $(U^n,V^n)$ are functions of $(J,K)=(J,M_2,S,S')$ given the codebook, and from the Fano's inequality $H(S'|S,J,M_2,Z^n) \leq n\epsilon_n$ (this is due to the codebook generation in which the size of $\mathcal{S}'$ for a given $(J,M_2,S)$ is less than $2^{nI(V;Z|U)}$ and therefore with high probability $S'$ can be decoded given $(S,J,M_2,Z^n)$), $(b)$ follows from bounding the term $H(U^n,V^n,Z^n)$ using properties of jointly typical sequences, i.e., $p(u^n,v^n,z^n) \leq 2^{-n(H(Z)+I(X;U,V|Z)-3\delta_{\epsilon})}=2^{-n(I(X;U,V)+H(Z|U,V)-3\delta_{\epsilon})}$, from the code construction, and from Lemma~\ref{lemma:1}, and $(c)$ from the code construction that $S \in [1:2^{n(I(Y;V|U)-I(Z;V|U)-\delta_{\epsilon})}]$.

\textbf{Converse}: $U_i$ and $V_i$ are defined as in the converse proof of Theorem \ref{theorem:region_exponent}. For any achievable $R_s$, it follows that
\begin{align*}
n(R_s-\delta_n) & \leq H(S)=H(S|M,Z^n)+I(S;M,Z^n)\\
&\overset{(a)}{\leq} H(S|M,Z^n)+n\delta_n\\
&\overset{(b)}{\leq} \sum_{i=1}^n I(V_i;Y_i|U_i)-I(V_i;Z_i|U_i)+n\delta_n+n\epsilon_n,
\end{align*}
where $(a)$ follows from the key leakage constraint and $(b)$ follows from the steps from $\eqref{eq:Econverse_start}$ to $\eqref{eq:Econverse_end}$.
\end{proof}

%%%%%%%%%%%%%%%%%%%%%%%%%%%%%%%%%%%%%%%%%%%%%%%%%%%%%%%%%%%%%%%%%%%%%%%%%%%%%%%%%%%%%%%%%%%%%%%%%%%%%%%%%%%%%%%%%%%

\appendices

\section{Converse Proof of Region $\mathcal{R}_{1,X-Y-Z}$} \label{appendix:converse_degraded_case}
Let $V_i \triangleq (M,S,Y_{i+1}^n,Z^{n\setminus i})$ which satisfies $V_i-X_i-Y_i-Z_i$ for all $i=1,\ldots,n$. For any achievable tuple $(R,L,E)$, it follows that
\begin{align*}
n(R + \delta_n) &\geq H(M) \\
&\geq H(M|Y^n,Z^n)-H(M,S|X^n,Y^n,Z^n)\\
&= H(M,S|Y^n,Z^n)-H(S|M,Y^n,Z^n)\\
&\qquad -H(M,S|X^n,Y^n,Z^n)\\
&\overset{(a)}{\geq}  I(M,S;X^n|Y^n,Z^n)-n\epsilon_n\\
&\overset{(b)}{\geq}  \sum_{i=1}^n I(X_i;V_i|Y_i)-n\epsilon_n,
\end{align*}
where $(a)$ follows from Fano's inequality $H(S|M,Y^n) \leq n\epsilon_n$ and $(b)$ follows from the Markov chain $X_i-Y_i-Z_i$, the definition of $V_i$, and that conditioning reduces entropy.

The information leakage,
\begin{align*}
&n(L + \delta_n)\\&\geq I(X^n;M,Z^n) = I(X^n;Z^n)+I(X^n;M|Z^n) \\
&\overset{(a)}{\geq} I(X^n;Z^n)+H(M|Z^n,Y^n)-H(M|X^n,Y^n,Z^n) \\
&\overset{(b)}{\geq} I(X^n;Z^n)+H(M,S|Z^n,Y^n)-n\epsilon_n \\ &\qquad -H(M,S|X^n,Y^n,Z^n) \\
&\overset{(c)}{\geq} \sum_{i=1}^n I(X_i;Z_i) + H(X_i|Y_i)-H(X_i|V_i,Y_i)-n\epsilon_n,
\end{align*}
where $(a)$ follows from the Markov chain $M-(X^n,Z^n)-Y^n$, $(b)$ follows from Fano's inequality, $(c)$ follows from the Markov chains $X_i-Y_i-Z_i$ and the definition of $V_i$.

The bound on mFAP exponent follows similarly as in the converse proof of Theorem \ref{theorem:region_exponent}, except that the steps from  \eqref{eq:Econverse_start} to \eqref{eq:Econverse_end} are replaced by
\begin{align*}
 H(S|M,Z^n)
&\overset{(a)}{\leq} H(S|M,Z^n)-H(S|M,Y^n)+n\epsilon_n \\
&\overset{(b)}{=} H(S|M,Z^n)-H(S|M,Y^n,Z^n)+n\epsilon_n \\
&\overset{(c)}{\leq} \sum_{i=1}^n H(Y_i|Z_i) -H(Y_i|V_i,Z_i)+n\epsilon_n,
\end{align*}
where $(a)$ from Fano's inequality, $(b)$ from the Markov chain $(S,M)-Y^n-Z^n$, and $(c)$ from the definition of $V_i$.

%%%%%%%%%%%%%%%%%%%%%%%%%%%%%%%%%%%%%%%%%%%%%%%%%%%%%%%%%%%%%%%%%%%%%%%%%%%%%%%%%%%%%%%%%%%%%%%%%%%%%%%%%%%%%%%%%%%
\section{Proof of Lemma \ref{lemma:1}} \label{appendix:Lemma1}
Let $E$ be a binary random variable taking value $0$ if $(X^n,U^n(J),V^n(J,K),Y^n,Z^n) \in \mathcal{T}_{\epsilon}^{(n)}$, and $1$ otherwise. Since $(X^n,U^n(J),V^n(J,K),Y^n,Z^n) \in \mathcal{T}_{\epsilon}^{(n)}$ with high probability, we have $\text{Pr}(E=1) \leq \delta_{\epsilon}$. It follows that
\begin{align*}
 H(Z^n|J)&\leq H(Z^n|U^n,E) + H(E)\\
  &\leq \text{Pr}(E=0) H(Z^n|U^n,E=0) \\ &\qquad+ \text{Pr}(E=1) H(Z^n|U^n,E=1) + h(\delta_{\epsilon})
 \end{align*}
 \begin{align*}
 &\leq H(Z^n|U^n,E=0) +\delta_{\epsilon} H(Z^n)+ h(\delta_{\epsilon}) \\
 &\leq H(Z^n|U^n,E=0) + n\delta_{\epsilon} \log|\mathcal{Z}| + h(\delta_{\epsilon})\\
&= \sum_{u^n \in \mathcal{T}_{\epsilon}^{(n)}} p(u^n|E=0) H(Z^n|U^n=u^n,E=0) \\ &\qquad + n\delta_{\epsilon} \log|\mathcal{Z}| + h(\delta_{\epsilon})\\
&\leq \sum_{u^n \in \mathcal{T}_{\epsilon}^{(n)}} p(u^n|E=0) \log|\mathcal{T}_{\epsilon}^{(n)}(Z|u^n)| + n\delta_{\epsilon} \log|\mathcal{Z}| \\&\qquad + h(\delta_{\epsilon})\leq n(H(Z|U)+\delta_{\epsilon}'),
\end{align*}
where $h(\cdot)$ is the binary entropy function, and the last inequality follows from the property of jointly typical set \cite{ElGamalKim} with $\delta_{\epsilon}, \delta_{\epsilon}' \rightarrow 0$ as $\epsilon \rightarrow 0$, and $\epsilon \rightarrow 0$ as $n \rightarrow \infty$.
%
%
%%%%%%%%%%%%%%%%%%%%%%%%%%%%%%%%%%%%%%%%%%%%%%%%%%%%%%%%%%%%%%%%%%%%%%%%%%%%%%%%%%%%%%%%%%%%%%%%%%%%%%%%%%%%%%%%%%%%
\section{Converse Proof of the mFAP Exponent Bound} \label{appendix:converse_mFAP}
Similarly as in \cite{WillemsIgnatenko}, let us define the set of secret key messages that can be reconstructed from $m$, i.e., $\mathcal{C}(m) = \{s: \text{there exists a sequence}\ y^n \in \mathcal{Y}^n\ \text{s.t.}\ g^{(n)}(m,y^n)=s\}$. Also, let $C(s,m) = 1$ for $s \in \mathcal{C}(m)$, and $0$ otherwise. We have that $\delta_{n}\geq \text{Pr}(\hat{S}\neq S) \geq \sum_{m}\text{Pr}(M=m,S \notin \mathcal{C}(m)) = \text{Pr}(C=0)$. An adversary who knows $m$ and $z^n$ can choose
a sequence $\tilde{y}^n$ that results in the MAP estimate, i.e.,
\begin{align}
\tilde{s}(m,z^n) = \arg\max_{s \in \mathcal{C}(m)}p(s|m,z^n), \label{eq:MAP_strategy_converse}
\end{align}
and achieves
\begin{align}
&\text{FAP} = \sum_{m,z^n}\text{Pr}(\tilde{s}=S,M=m,Z^n=z^n)\nonumber \\
&\overset{(a)}{=}\sum_{m,z^n}p(m,z^n)\max_{s \in \mathcal{C}(m)}p(s|m,z^n)\nonumber\\
&\geq \sum_{m,z^n}p(m,z^n)\max_{s \in \mathcal{C}(m)}p(s,C=1|m,z^n)\nonumber\\
&\geq \sum_{m,z^n}p(m,z^n)p(C=1|m,z^n)\max_{s \in \mathcal{C}(m)}p(s|m,z^n,C=1), \label{eq:FAP_converse}
\end{align}
where $(a)$ follows from $\eqref{eq:MAP_strategy_converse}$.
Then for any achievable $E$, it follows that
\begin{align*}
&n(E-\delta_n) \leq \log\Big(\frac{1}{\text{mFAP}}\Big) \leq \log\Big(\frac{1}{\text{FAP}}\Big) \\
&\overset{(a)}{\leq} -\log \big(\text{Pr}(C=1)\big) \\ &\qquad -\log\big(\sum_{m,z^n}p(m,z^n|C=1)\max_{s\in \mathcal{C}(m)}p(s|m,z^n,C=1)\big) \\
&\overset{(b)}{\leq} -\log(1-\delta_n) \\ &\qquad -\sum_{m,z^n}p(m,z^n|C=1)\log\big(\max_{s\in \mathcal{C}(m)}p(s|m,z^n,C=1)\big) \\
&\leq -\log(1-\delta_n) -\sum_{m,z^n}p(m,z^n|C=1)\cdot  \\ &\qquad \sum_{s\in \mathcal{C}(m)}p(s|m,z^n,C=1)\log(p(s|m,z^n,C=1))\\
&= -\log(1-\delta_n)+ H(S|M,Z^n,C=1),
\end{align*}
where $(a)$ follows from $\eqref{eq:FAP_converse}$ and $(b)$ follows from $\text{Pr}(C=1)\geq 1-\delta_n$ and Jensen's inequality
\cite{Williams}.

Continuing the chain of inequalities where $(1-\delta_n)H(S|M,Z^n,C=1)\leq \text{Pr}(C=1) H(S|M,Z^n,C=1) \leq H(S|M,Z^n)$, we get
\begin{align}
& (1-\delta_n)\cdot[n(E-\delta_n)+\log(1-\delta_n)] \nonumber \\
&\leq H(S|M,Z^n)\label{eq:Econverse_start}\\
&\overset{(a)}{\leq} H(S|M,Z^n)-H(S|M,Y^n)+n\epsilon_n \nonumber\\
%&= I(S;Y^n|M)-I(S;Z^n|M)+n\epsilon_n \nonumber\\
&= \sum_{i=1}^n I(S;Y_i|M,Y_{i+1}^n)-I(S;Z_i|M,Z^{i-1})+n\epsilon_n \nonumber\\
&\overset{(b)}{=} \sum_{i=1}^n I(S,Z^{i-1};Y_i|M,Y_{i+1}^n)-I(S,Y_{i+1}^n;Z_i|M,Z^{i-1})\nonumber \\ &\qquad +n\epsilon_n \nonumber\\
&\overset{(c)}{=} \sum_{i=1}^n I(S;Y_i|M,Y_{i+1}^n,Z^{i-1})-I(S;Z_i|M,Y_{i+1}^n,Z^{i-1})\nonumber \\ &\qquad +n\epsilon_n \nonumber\\
&\overset{(d)}{=} \sum_{i=1}^n I(V_i;Y_i|U_i)-I(V_i;Z_i|U_i)+n\epsilon_n, \label{eq:Econverse_end}
\end{align}
where $(a)$ follows from Fano's inequality, and $(b)$ and $(c)$ from the Csisz\'{a}r's sum identity $\sum_{i=1}^n I(Z^{i-1};Y_i|M,S,Y_{i+1}^n)-I(Y_{i+1}^n;Z_i|M,S,Z^{i-1})=0=\sum_{i=1}^n I(Z^{i-1};Y_i|M,Y_{i+1}^n)-I(Y_{i+1}^n;Z_i|M,Z^{i-1})$, and $(d)$ from the definitions $U_i \triangleq (M,Y_{i+1}^n,Z^{i-1})$ and $V_i \triangleq (M,S,Y_{i+1}^n,Z^{i-1})$.

%%%%%%%%%%%%%%%%%%%%%%%%%%%%%%%%%%%%%%%%%%%%%%%%%%%%%%%%%%%%%%%%%%%%%%%%%%%%%%%%%%%%%%%%%%%%%%%%%%%%%%%%%%%%%%%%%%%
\section{Cardinality Bounds of The Sets $\mathcal{U}$ and $\mathcal{V}$ in Theorem \ref{theorem:region_exponent}} \label{appendix:cardinality}
Consider the expression of $\mathcal{R}_{1}$ in Theorem \ref{theorem:region_exponent}:
\begin{align*}
R &\geq I(X;V|Y), \\
L & \geq I(X;V,Y) - I(X;Y|U)+I(X;Z|U),\\
E &\leq I(V;Y|U)-I(V;Z|U),
\end{align*}
for some $U \in \mathcal{U}$, $V \in \mathcal{V}$ such that $U-V-X-(Y,Z)$ forms a Markov chain.

We can rewrite some mutual information terms in the expression above as
\begin{align*}
     R &\geq H(X|Y)-H(X,Y|V)+H(Y|V),\\
     L &\geq H(X)-H(X,Y|V)+H(Y|V) -H(Y|U)+ H(Y|X)\\&\qquad+ H(Z|U)-H(Z|X),\\
     E &\leq H(Y|U)-H(Y|V)-H(Z|U)+H(Z|V).
\end{align*}

We will show that the random variables $U$ and $V$ may be replaced by new ones, satisfying $|\mathcal{U}| \leq  |\mathcal{X}|+3$, $|\mathcal{V}| \leq  (|\mathcal{X}|+3)(|\mathcal{X}|+2)$, and preserving the terms $H(X,Y|V),H(Y|V),H(Z|V)$, and $H(Y|U)-H(Z|U)$.

First, we bound the cardinality of the set $\mathcal{U}$.
Let us define the following $|\mathcal{X}|+3$ continuous functions of $p(v|u)$, $v \in \mathcal{V}$,
\begin{align*}
&f_{j}(p(v|u)) = \sum_{v \in \mathcal{V}}p(v|u)p(x|u,v),\ j=1,\ldots,|\mathcal{X}|-1, \\
&f_{|\mathcal{X}|}(p(v|u))  = H(X,Y|V,U=u)\\ & \qquad \qquad \qquad = H(X,Y,V|U=u)-H(V|U=u),\\
&f_{|\mathcal{X}|+1}(p(v|u)) =H(Y|V,U=u)\\ & \qquad \qquad \qquad = H(Y,V|U=u)-H(V|U=u),\\
&f_{|\mathcal{X}|+2}(p(v|u)) =H(Z|V,U=u)\\ & \qquad \qquad \qquad = H(Z,V|U=u)-H(V|U=u),\\
&f_{|\mathcal{X}|+3}(p(v|u)) =H(Y|U=u)-H(Z|U=u).
\end{align*}
The corresponding averages are
\begin{align*}
& \sum_{u \in \mathcal{U}}p(u)f_{j}(p(v|u))=P_{X}(x),\ j=1,\ldots,|\mathcal{X}|-1, \\
& \sum_{u \in \mathcal{U}}p(u) f_{|\mathcal{X}|}(p(v|u))= H(X,Y,V|U)-H(V|U), \\
& \sum_{u \in \mathcal{U}}p(u) f_{|\mathcal{X}|+1}(p(v|u))=  H(Y,V|U)-H(V|U), \\
& \sum_{u \in \mathcal{U}}p(u) f_{|\mathcal{X}|+2}(p(v|u))=  H(Z,V|U)-H(V|U), \\
& \sum_{u \in \mathcal{U}}p(u) f_{|\mathcal{X}|+3}(p(v|u))= H(Y|U)-H(Z|U).
\end{align*}
According to the support lemma \cite{CsiszarBook}, we can deduce that there exists a new random variable $U'$ jointly distributed with $(X,Y,Z,V)$ whose alphabet size is $|\mathcal{U}'|= |\mathcal{X}|+3$, and numbers $\alpha_{i} \geq 0$ with $\sum_{i=1}^{|\mathcal{X}|+3}\alpha_{i} =1$ that satisfy
\begin{align*}
&\sum_{i=1}^{|\mathcal{X}|+3}\alpha_{i} f_{j}(P_{V|U'}(v|i)) = P_{X}(x),\ j=1,\ldots,|\mathcal{X}|-1, \\
&\sum_{i=1}^{|\mathcal{X}|+3}\alpha_{i}f_{|\mathcal{X}|}(P_{V|U'}(v|i)) = H(X,Y,V|U')-H(V|U'),\\
&\sum_{i=1}^{|\mathcal{X}|+3}\alpha_{i}f_{|\mathcal{X}|+1}(P_{V|U'}(v|i)) = H(Y,V|U')-H(V|U'),\\
&\sum_{i=1}^{|\mathcal{X}|+3}\alpha_{i}f_{|\mathcal{X}|+2}(P_{V|U'}(v|i)) = H(Z,V|U')-H(V|U'),\\
&\sum_{i=1}^{|\mathcal{X}|+3}\alpha_{i}f_{|\mathcal{X}|+3}(P_{V|U'}(v|i)) = H(Y|U')-H(Z|U').
\end{align*}
Note that we have
\begin{align*}
&H(X,Y,V|U')-H(V|U') \\
&= H(X,Y,V|U)-H(V|U)\\
& \overset{(a)}{=} H(X,Y|V),
\end{align*}
where $(a)$ follows from the Markov chain $U-V-X-(Y,Z)$.
Similarly, from the Markov chain $U-V-X-(Y,Z)$, we have that $H(Y,V|U')-H(V|U')=H(Y,V|U)-H(V|U)=H(Y|V)$, and $H(Z,V|U')-H(V|U')=H(Z,V|U)-H(V|U)=H(Z|V)$.
Since $P_{X}(x)$ is preserved, $P_{X,Y,Z}(x,y,z)$ is also preserved.  Thus, $H(X|Y), H(Y|X), H(Z|X)$ are preserved.

Next we bound the cardinality of the set $\mathcal{V}$.
For each $u' \in \mathcal{U}'$, we define the following $|\mathcal{X}|+2$ continuous functions of $p(x|u',v)$,\ $x \in \mathcal{X}$,
\begin{align*}
&f_{j}(p(x|u',v)) = p(x|u',v),\ j=1,\ldots,|\mathcal{X}|-1, \\
&f_{|\mathcal{X}|}(p(x|u',v)) = H(X,Y|U'=u',V=v),\\
&f_{|\mathcal{X}|+1}(p(x|u',v)) = H(Y|U'=u',V=v),\\
&f_{|\mathcal{X}|+2}(p(x|u',v)) = H(Z|U'=u',V=v).
\end{align*}
Similarly to the previous part in bounding $|\mathcal{U}|$, there exists a new random variable $V'|\{U'=u'\} \sim p(v'|u')$ such that $|\mathcal{V}'| =  |\mathcal{X}|+2$ and $p(x|u')$, $H(X,Y|U'=u',V)$, $H(Y|U'=u',V)$, and $H(Z|U'=u',V)$ are preserved.

By setting $V'' =(V',U')$ where $\mathcal{V}'' = \mathcal{V}' \times \mathcal{U}'$, we have that $U'-V''-X-(Y,Z)$ forms a Markov chain.

Furthermore, we have the following preservations by $V''$,
\begin{align*}
&H(X,Y|V'') \\
& = H(X,Y|V',U')\\
& \overset{(a)}{=} H(X,Y|V,U')\\
& \overset{(b)}{=} H(X,Y|V,U)\\
& \overset{(c)}{=} H(X,Y|V),
\end{align*}
where $(a)$ follows from preservation by $V'$, $(b)$ follows from preservation by $U'$, and $(c)$ follows from the Markov chain $U-V-X-(Y,Z)$.
Similarly, from preservation by $U'$ and $V'$, and the Markov chain $U-V-X-(Y,Z)$, we have that $H(Y|V'') =H(Y|V',U')=H(Y|V)$ and $H(Z|V'')=H(Z|V',U')=H(Z|V)$.

 Therefore, we have shown that $U \in \mathcal{U}$ and $V \in \mathcal{V}$ may be replaced by $U' \in \mathcal{U}'$ and $V'' \in \mathcal{V}''$ satisfying
 \begin{align*}
 |\mathcal{U}'|&= |\mathcal{X}|+3, \\
 |\mathcal{V}''| &= |\mathcal{U}'||\mathcal{V}'|= (|\mathcal{X}|+3)(|\mathcal{X}|+2),
 \end{align*}
and  preserving the terms $H(X,Y|V),H(Y|V),H(Z|V)$, and $H(Y|U)-H(Z|U)$.

%%%%%%%%%%%%%%%%%%%%%%%%%%%%%%%%%%%%%%%%%%%%%%%%%%%%%%%%%%%%%%%%%%%%%%%%%%%%%%%%%%%%%%%%%%%%%%%%%%%%%%%%%%%%%%%%%%%
\section{Proof of the Compression-leakage-mFAP Exponent Region in the Binary Example} \label{appendix:proof_example}
\emph{Achievability:}  Let $V$ be an output of a BSC($\alpha$) with input $X$. Then it follows from the expression of $\mathcal{R}_{1,X-Y-Z}$ that
\begin{align*}
 R &\geq I(X;V|Y) \\
  &\overset{(a)}{=}p\cdot(H(X)-H(X|V)) \\
  &\overset{(b)}{=}p\cdot(1-h(\alpha)),
\end{align*}
where $(a)$ follows since $Y=e$ with probability $p$, otherwise $Y=X$, and $(b)$ follows from the choice of $V$,
\begin{align*}
  L &\geq I(X;Z)+ I(X;V|Y)\\
  &\overset{(a)}{=} 1-H(X|Z) + p\cdot(1-h(\alpha))\\
  &\overset{(b)}{=} 1-((1-p)q +p)+ p\cdot(1-h(\alpha))\\
  &= (1-q)(1-p) + p\cdot(1-h(\alpha)),
\end{align*}
where $(a)$ follows from the bound on $R$ and $(b)$  follows since $Z=e$ with probability $(1-p)q +p$, otherwise $Z=X$.
\begin{align*}
   E &\leq I(Y;V|Z)\\
 &\overset{(a)}{=}  I(X;V|Z)-I(X;V|Y)\\
 &\overset{(b)}{=} ((1-p)q +p)\cdot I(X;V)- p\cdot(1-h(\alpha))\\
 &= q (1-p)(1-h(\alpha)),
\end{align*}
where  $(a)$  follows from the Markov chain $V-X-Y-Z$ and $(b)$ follows since $Z=e$ with probability $(1-p)q +p$, otherwise $Z=X$.

\emph{Converse:}
Let $(R,L,E)$ be an achievable tuple. We now prove that there exist $\alpha \in [0,1/2]$ satisfying the inequalities shown in the achievability above. From $\mathcal{R}_{1,X-Y-Z}$, we have the following bound on the compression rate $R$.
\begin{align*}
  R &\geq I(X;V|Y) \\
  &= p\cdot I(X;V)  \\
  & = p\cdot(1-H(X|V)).
\end{align*}
Since $0 \leq H(X|V) \leq H(X)=1$, and $h(\cdot)$ is a continuous one-to-one mapping from $[0,1/2]$ to $[0,1]$, there exists $\alpha \in [0,1/2]$ s.t. $H(X|V)=h(\alpha)$, and thus $R \geq p\cdot(1-h(\alpha))$. The bounds on $L$ and $E$ readily follow from $H(X|V)=h(\alpha)$.

\end{document}